\documentclass[10pt, conference, letterpaper]{IEEEtran}
\IEEEoverridecommandlockouts
\usepackage{url}
\usepackage{amssymb}
\usepackage{amsfonts}
\usepackage{amsthm}
\usepackage{amsmath,amssymb,verbatim}
\usepackage{graphicx}
\usepackage{subcaption}
\usepackage{float}
\usepackage{epsfig}
\usepackage{threeparttable}
\usepackage{algorithm}
\usepackage{algorithmic}
\usepackage{color} 
\usepackage{dblfloatfix}
\usepackage{changepage}
\usepackage{comment}
\usepackage{multirow}
\usepackage{colortbl}
\usepackage{ulem}
\usepackage{color}
\usepackage{lipsum}
\usepackage{cite}
\usepackage{tabularx}
\usepackage{dsfont}

\floatname{algorithm}{Algorithm}

\newcommand{\blue}{\textcolor{black}}

\newtheorem{theorem}{\it Theorem}

\newtheorem{corollary}{\it Corollary}

\newcommand{\bE}{\mathds{E}}

\DeclareMathOperator*{\argmax}{argmax}
\DeclareMathOperator*{\argmin}{argmin}
\begin{document}

\title{The Power of Waiting for More than One Response in Minimizing the Age-of-Information}

\author{
Yu Sang,
Bin Li,
and Bo Ji 
\thanks{
This work was supported in part by the NSF under Grants CCF-1657162, CNS-1651947, and CNS-1717108.

Yu Sang (yu.sang@temple.edu) and Bo Ji (boji@temple.edu) are with the Department of Computer and
Information Sciences, Temple University, Philadelphia, PA, and Bin Li (binli@uri.edu) is with the Department of Electrical, Computer and Biomedical Engineering,
University of Rhode Island, Kingston, Rhode Island.}
}
\maketitle

\begin{abstract}
The Age-of-Information (AoI) has recently been proposed as an important metric for investigating the timeliness performance in information-update systems. Prior studies on AoI optimization often consider a \emph{Push} model, which is concerned about when and how to ``push" (i.e., generate and transmit) the updated information to the user. In stark contrast, in this paper we introduce a new \emph{Pull} model, which is more relevant for certain applications (such as the real-time stock quotes service), where a user sends requests to the servers to proactively ``pull" the information of interest. Moreover, we propose to employ request replication to reduce the AoI. Interestingly, we find that under this new Pull model, replication schemes capture a novel tradeoff between different levels of information freshness and different response times across the servers, which can be exploited to minimize the expected AoI at the user's side. Specifically, assuming Poisson updating process at the servers and exponentially distributed response time, we derive a closed-form formula for computing the expected AoI and obtain the optimal number of responses to wait for to minimize the expected AoI. Finally, we conduct numerical simulations to elucidate our theoretical results. Our findings show that waiting for more than one response can significantly reduce the AoI in most scenarios.
%Researchers have proposed using redundant requests to improve the performances, such as delay and availability, of communication networks.
%In this paper, we study the problem of using redundancy to optimize a new metric, Age of Information (AoI).
%We consider a real-time system that the client pulls time-sensitive data from servers.
%Different from traditional redundancy mechanism, the client may need multiple responses to minimize the AoI.
%We show that due to this unique feature, redundancy has special impacts on the network performance.
%We provide quantitative analysis of the impact of redundancy on AoI.
%Then, we solve the optimization problem of minimizing AoI using redundant requests and find the closed form optimal solutions under the exponential distribution assumption.
%Numerical simulations are performed to verify the correctness of our theoretical results.
\end{abstract}

\section{Introduction}\label{sec:intro}
The last decades have witnessed the prevalence of smart devices and significant advances in ubiquitous computing and the Internet of things.
This trend is forecasted to continue in the years to come \cite{vni}. 
The development of this trend has spawned a plethora of real-time services that require timely information/status updates. 
One practically important example of such services is vehicular networks and intelligent transportation systems~\cite{kaul11secon,kaul11globecom}, where accurate status information (position, speed, acceleration, tire pressure, etc.) of a vehicle needs to be shared with other nearby vehicles and road-side facilities in a timely manner in order to avoid collisions and ensure substantially improved road safety.
More such examples include sensor networks for environment/health monitoring~\cite{ko10,corke10}, wireless channel feedback~\cite{costa15}, news feeds, weather updates, online social networks, fare aggregating sites (e.g., Google Shopping), and stock quotes service. 

For systems providing such real-time services, those commonly used performance metrics, such as throughput and delay, 
exhibit significant limitations in measuring the system performance \cite{kaul12infocom}. 
Instead, \emph{the timeliness of information updates becomes a major concern}. To that end, a new metric called the \emph{Age-of-Information (AoI)} 
has been proposed as an important metric for studying the timeliness performance \cite{kaul11secon}.
The AoI is defined as the time elapsed since the most recent update occurred (see Eq.~\eqref{eq:aoi} for a formal definition).
Using the AoI metric introduced in \cite{kaul11secon} for vehicular networks, the work of \cite{kaul12infocom} 
employs a simple system model to analyze and optimize the timeliness performance of an information-update system.
This seminal work has recently aroused dramatic interests from the research community and has inspired a series
of interesting studies on the AoI analysis and optimization (see \cite{sun17update}  and references therein).

While all prior studies consider a \emph{Push} model, which is concerned about when and how to 
``push" (i.e., generate and transmit) the updated information to the user, in this paper we introduce a new 
\emph{Pull} model, under which a user sends requests to the servers to proactively ``pull" the information of interest.  
This Pull model is more relevant for many important applications where the user's interest is in the freshness 
of information at the point when the user requests it rather than in continuously monitoring the freshness of information. 
One application of the Pull model is in the real-time stock quotes service, where a customer (i.e., user) submits a query 
to multiple stock quotes providers (i.e., servers) and each provider responds with the most up-to-date information it has.
%Here, each provider can be viewed as a server. 
%Based on the Pull model, \emph{we propose to employ replication schemes, which capture 
%a new tradeoff between the freshness at the sources and the response times, to minimize the average
%age over all the requests.}

\emph{To the best of our knowledge, however, none of the existing work on the timeliness optimization has considered 
such a Pull model. In stark contrast, we focus on the Pull model and propose to employ request replication 
to minimize the expected AoI at the user's side.} Although a similar Pull model is considered for data synchronization 
in \cite{bright04,bright06}, the problems are quite different and request replication is not exploited. 
Note that the concept of replication is not new and has been extensively studied for various applications 
(e.g., cloud computing and datacenters~\cite{gardner15,ananthanarayanan12}, 
storage clouds~\cite{li16}, parallel computing~\cite{wang14,wang15}, 
and databases \cite{pacitti99,pereira10}).
\emph{However, for the AoI minimization problem under the Pull model, replication schemes exhibit a 
unique property and capture a novel tradeoff between different levels of information freshness and different 
response times across the servers. This tradeoff reveals the power of waiting for more than one response 
and can be exploited to minimize the expected AoI at the user's side.}

Next, we explain the above key tradeoff through a comparison with cloud computing systems.
It has been observed that in a cloud or a datacenter, the processing time of a same job can be highly 
variable on different servers \cite{ananthanarayanan12}.
Due to this important fact, replicating a job on multiple servers and waiting for the first finished copy can
help reduce latency \cite{ananthanarayanan12,gardner15}. Apparently, in such a system it is \emph{not} 
beneficial to wait for more copies of the job to finish, as all the copies would give the same outcome. 
In contrast, in the information-update system we consider, although the servers may possess the same 
type of information (weather forecast, stock prices, etc.), they could have different versions of the 
information with different levels of freshness due to the random updating processes.
Hence, the first response may come from a server with stale information; waiting for more than one 
response has the potential of receiving fresher information and thus helps reduce the AoI. 
\emph{Hence, it is no longer the best to stop after receiving the first response (as in the other aforementioned applications).}
On the other hand, waiting for too many responses will lead to a longer total waiting time and thus, also incurs 
a larger AoI at the user's side.
\emph{Therefore, it is challenging to determine the optimal number of responses to wait for in order to minimize 
the expected AoI at the user's side.} 

In what follows, we summarize the key contributions of this paper. \emph{First}, for the first time we introduce the Pull 
model for studying the timeliness optimization problem and propose to employ request replication to reduce the AoI.
\emph{Second}, assuming Poisson updating process at the servers and exponentially distributed response time, 
we derive a closed-form formula for computing the expected AoI and obtain the optimal number of responses to wait 
for to minimize the expected AoI. Some extensions are also discussed.
\emph{Third}, we conduct extensive numerical simulations to elucidate our theoretical results. We also investigate 
the impact of the system parameters (the updating rate, the mean response time, and the total number of servers) on the achieved 
gain in the AoI reduction. Simulation results for other types of response time distribution are also provided.
Our findings show that waiting for more than one response can significantly reduce the AoI in most scenarios.

The remainder of this paper is organized as follows. We first describe our new Pull model in Section~\ref{sec:model}. 
Then, we analyze the expected AoI under replication schemes in Section~\ref{sec:e_aoi}, obtain the
optimal number of responses for minimizing the expected AoI in Section~\ref{sec:opt}, and briefly discuss some 
extensions of our work in Section~\ref{sec:extensions}. Section~\ref{sec:simulation} presents the simulation results. 
Finally, we conclude the paper in Section~\ref{sec:conclusion}.

\section{System Model}\label{sec:model}
We consider an information-update system where a user pulls time-sensitive information from $n$ servers. 
These $n$ servers are connected to a common information source and update their data \emph{asynchronously}. 
We call such a model the \emph{Pull} model (see Fig.~\ref{fig:model}). Let $i \in \{1,2,\dots, n\}$ be the server index.
We assume that the information updating process at server $i$ is Poisson with rate $\lambda>0$ and is independent
and identically distributed (\emph{i.i.d.}) across the servers. 
This implies that the inter-update time (i.e., the time duration between two successive updates) at each server 
follows an exponential distribution with mean $1/\lambda$. Here, the inter-update time at a server can be 
interpreted as the time required for the server to receive information updates from the source.
Let $u_i(t)$ denote the time when the most recent update at server $i$ occurs, and let $\Delta_i(t)$ denote the 
AoI at server $i$, which is defined as the time elapsed since the most recent update at this server:
\begin{equation}
\label{eq:aoi}
\Delta_i(t) \triangleq t - u_i(t). 
\end{equation}
Therefore, if an update occurs at a server, then the AoI at this server drops to zero; otherwise, the AoI increases 
linearly as time goes by until the next update occurs. Fig.~\ref{fig:update} provides an illustration of the AoI evolution at server $i$.

\begin{figure}[!t]
\centering
\includegraphics[width=0.4\textwidth]{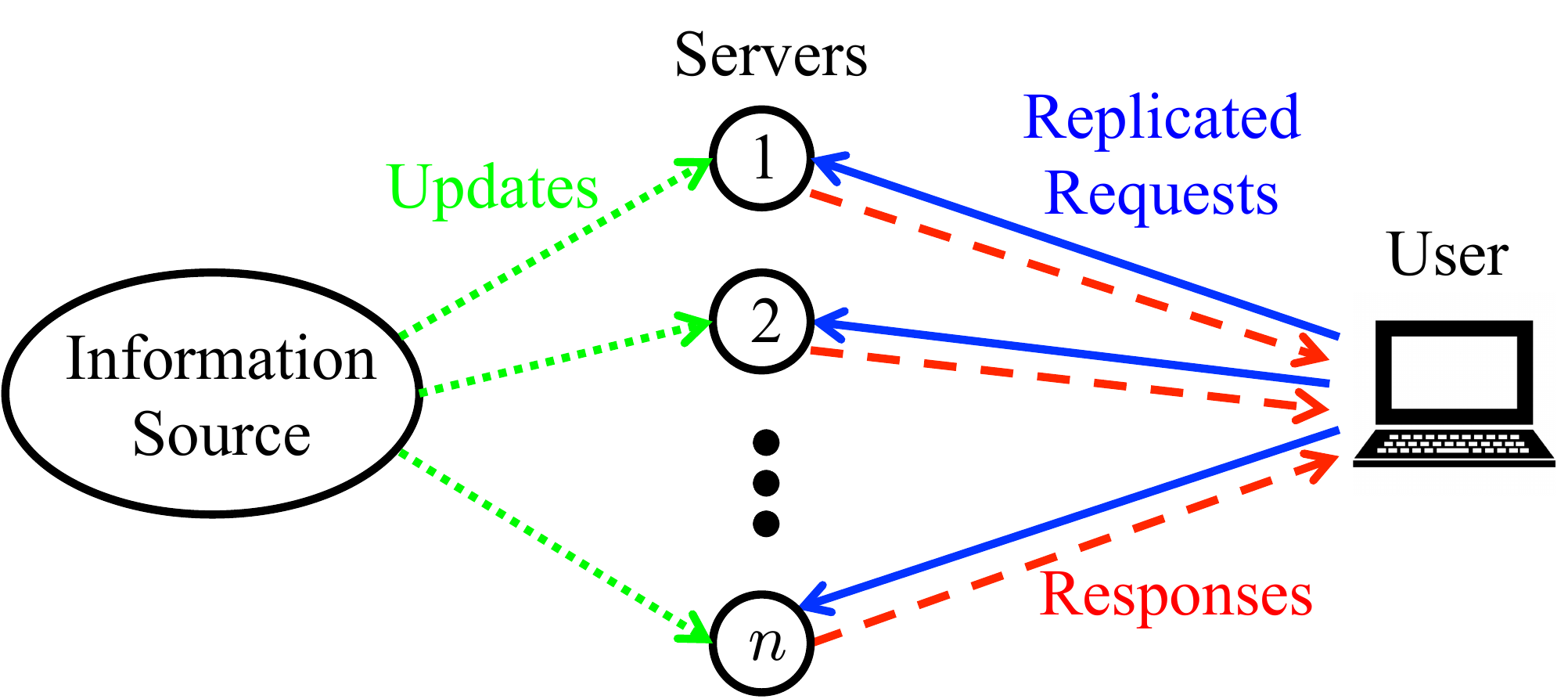}
\caption{The Pull model of information-update systems. 
\blue{Note that the arrows in the figure denote logical links rather than physical connections. 
The updates, requests, and responses are all transmitted through (wired or wireless) networks.}}
\label{fig:model}
\end{figure}

\begin{figure}[!t]
\centering
\includegraphics[width=0.35\textwidth]{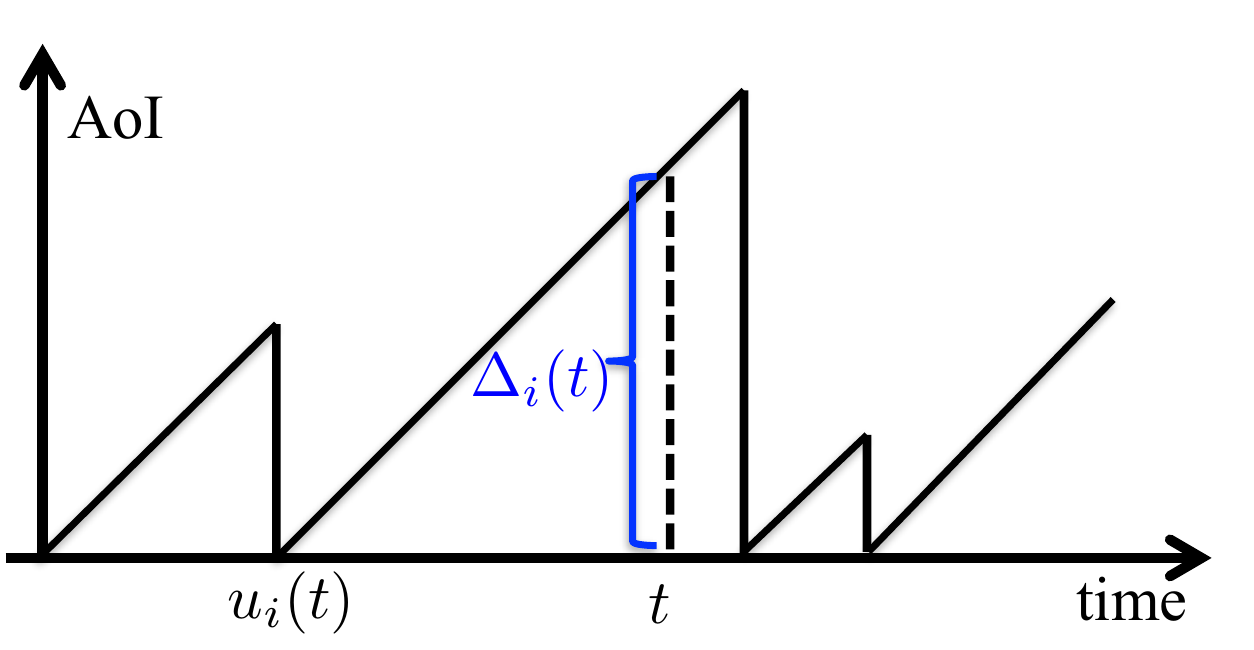}
\caption{An illustration of the AoI evolution at server $i$.}
\label{fig:update}
\end{figure}

In this work, we consider the $(n,k)$ replication scheme, under which the user sends the replicated copies of the 
request to all $n$ servers and waits for the first $k$ responses. Let $R_i$ denote the response time for server $i$.
Note that each server may have a different response time, which is the time elapsed since the request is sent out 
by the user until the user receives the response from this server. We assume that the time for the requests to reach 
the servers is negligible compared to the time for the user to download the data from the servers. Hence, the response 
time can be interpreted as the downloading time. Let $s$ denote the downloading start time, which is the same for 
all the servers, and let $f_i$ denote the downloading finish time for server $i$. Then, the response time for server $i$ 
is $R_i = f_i - s$. We assume that the response time is exponentially distributed with mean $1/\mu$ and is \emph{i.i.d.}  
across the servers.
\blue{Note that the model we consider above is simple, but it suffices to capture the key aspects and novelty of the problem we study.}

Under the $(n,k)$ replication scheme, when the user receives the first $k$ responses, it uses the freshest information 
among these $k$ responses to make certain decisions (e.g., stock trading decisions based on the received stock price 
information). Let $(j)$ denote the index of the server corresponding to the $j$-th response received by the user. 
Then, set $K=\{(1), (2), \dots, (k)\}$ contains the indices of the servers that return the first $k$ responses,
and the following is satisfied: $f_{(1)} \le f_{(2)} \le \dots \le f_{(k)}$ and $R_{(1)} \le R_{(2)} \le \dots \le R_{(k)}$.
Let server $i^*$ be the one that contains the freshest information (i.e., that has the smallest AoI) among these $k$ responses
when downloading starts at time $s$, i.e., $i^* = \argmin_{i \in K} \Delta_i(s)$ (or $i^* = \argmax_{i \in K} u_i(s)$ due to Eq.~\eqref{eq:aoi}).
Here, we are interested in the AoI at the user's side when it receives the $k$-th response, denoted by $\Delta(k)$, 
which is the time difference between when the $k$-th response is received and when the information at 
server $i^*$ is updated, i.e., 
\begin{equation}
\label{eq:aoik}
\Delta(k) \triangleq f_{(k)} - u_{i^*}(s).
\end{equation} 

Then, there are two natural questions of interest. \emph{First, for a given $k$, can one obtain a closed-form formula 
for computing the expected AoI at the user's side, $\bE[\Delta(k)]$? Second, how to determine the optimal number of responses 
to wait for, such that $\bE[\Delta(k)]$ is minimized?} The second question can be formulated as the following optimization problem:
\begin{equation}
\label{eq:opt}
\min_{k \in \{1,2,\dots, n\}}  \bE\left[\Delta(k)\right].
\end{equation} 
We will answer these two questions in the following two sections, respectively.

\section{Expected AoI}\label{sec:e_aoi}
In this section, we focus on answering the first question and derive a closed-form formula 
for computing the expected AoI at the user's side under the $(n,k)$ replication scheme. 

We begin with the definition of $\Delta(k)$ and rewrite Eq.~\eqref{eq:aoik} as follows: 
\begin{equation}
%\[
\begin{split}
\label{eq:aoik_l}
\Delta(k) =& f_{(k)} - u_{i^*}(s) \\
=& f_{(k)} - s + s -  u_{i^*}(s) \\
\stackrel{(a)} =& R_{(k)} + s - \max_{i \in K} u_i(s) \\
=& R_{(k)} + \min_{i \in K} \{s - u_i(s)\} \\
\stackrel{(b)} =& R_{(k)} + \min_{i \in K} \Delta_i(s),
\end{split}
%\]
\end{equation} 
where (a) is from the definition of $R_i$ and $i^*$ and (b) is from the definition of $\Delta_i(t)$ (i.e., Eq.~\eqref{eq:aoi}).
As can be seen from the above expression, under the $(n,k)$ replication scheme 
the AoI at the user's side consists of two terms:
(i) $R_{(k)}$, the total waiting time for receiving the first $k$ responses,  
and (ii) $\min_{i \in K} \Delta_i(s)$ (or $\Delta_{i^*}(s)$), the AoI of the freshest information 
among these $k$ responses when downloading starts at time $s$.
An illustration of these two terms and $\Delta(k)$ is shown in Fig.~\ref{fig:tradeoff}.

Taking the expectation of both sides of Eq.~\eqref{eq:aoik_l}, we have

\begin{equation}
\label{eq:e_aoi}
\bE[\Delta(k)] = \bE \left[R_{(k)} \right] + \bE \left [\min_{i \in K} \Delta_i(s) \right].
\end{equation}
Intuitively, as $k$ increases, i.e., waiting for more responses, the expected total waiting time
(i.e., the first term) increases. On the other hand, upon receiving more responses, the expected 
AoI of the freshest information among these $k$ responses (i.e., the second term) decreases. 
Hence, there is a natural tradeoff between these two terms, which is a unique property of our 
newly introduced Pull model. 

Next, we formalize this tradeoff by deriving the closed-form expressions of the above two terms
as well as the expected AoI.
We state the main result of this section in Theorem~\ref{thm:e_aoi}.

\begin{figure}[!t]
\centering
\includegraphics[width=0.4\textwidth]{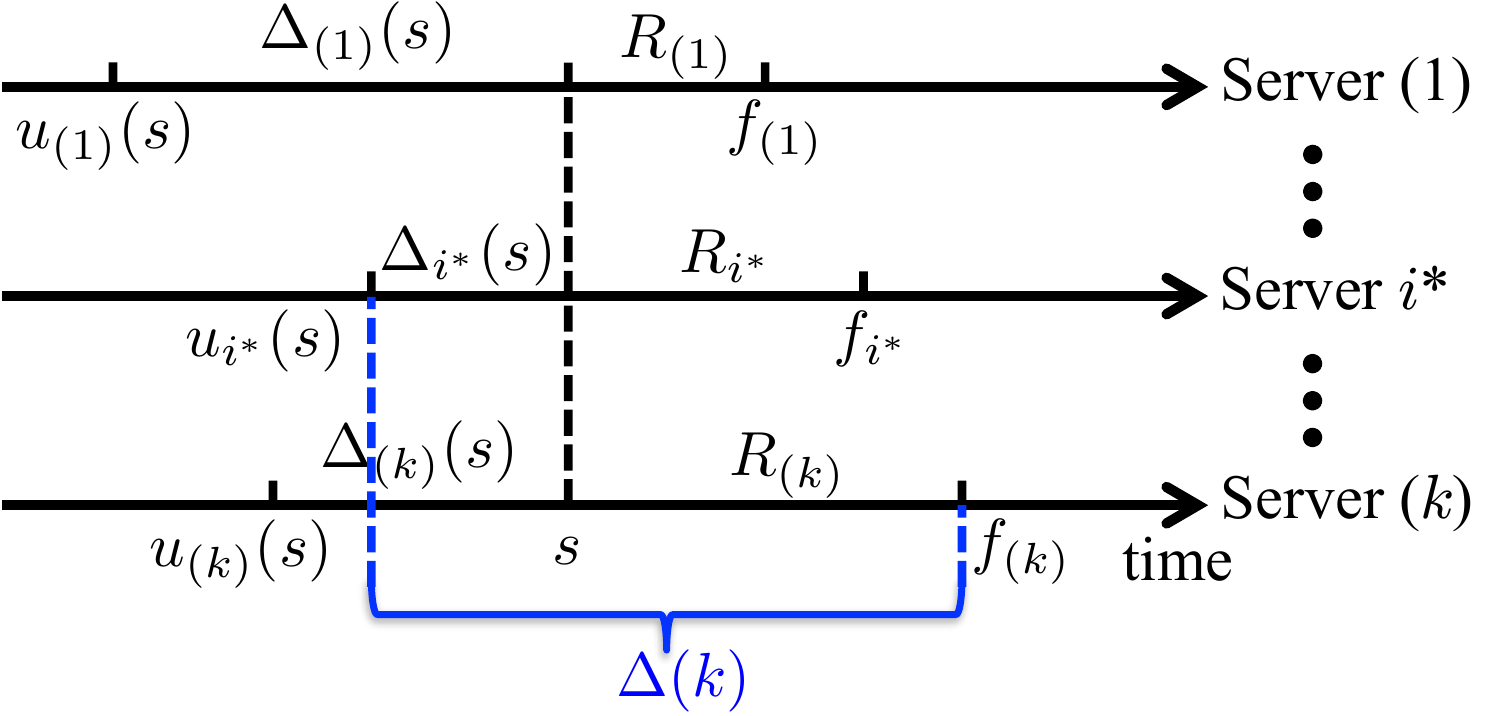}
\caption{An illustration of the AoI at the user's side and its two terms under the ($n,k$) replication scheme.}
\label{fig:tradeoff}
\end{figure}

\begin{theorem}\label{thm:e_aoi}
Under the $(n,k)$ replication scheme, the expected AoI at the user's side can be expressed as:
\begin{equation}
\label{eq:e_aoi_formula}
\bE[\Delta(k)] = \frac{1}{\mu}(\mathbf{H}(n) - \mathbf{H}(n-k)) + \frac{1}{k\lambda},
\end{equation}
where $\mathbf{H}(n) = \sum_{l=1}^n \frac{1}{l}$ is the $n$-th partial sum of the diverging harmonic series.
\end{theorem}

\begin{proof}
\blue{We first analyze the the first term of the right-hand side of Eq.~\eqref{eq:e_aoi} and want to show 
$\bE[R_{(k)}] = \frac{1}{\mu}(\mathbf{H}(n) - \mathbf{H}(n-k))$.}
Note that the response time is exponentially distributed with mean $1/\mu$ and is \emph{i.i.d.} across the servers.
Hence, random variable $R_{(k)}$ is the $k$-th smallest value of $n$ \emph{i.i.d.} exponential random variables with mean $1/\mu$.
The order statistics results of exponential random variables give that $R_{(1)}$ is an exponential random variable with mean $\frac{1}{n\mu}$
and that $(R_{(j)} - R_{(j-1)})$ is an exponential random variable with mean $\frac{1}{(n+1-j)\mu}$ for any $j \in \{2,3,\dots,n\}$~\cite{firstcourse}.
Hence, we have the following:
\begin{align}
\bE \left[R_{(k)} \right] 
&= \bE \left[ R_{(1)} + \sum_{j=2}^{k} (R_{(j)} - R_{(j-1)}) \right]\notag\\
&= \bE[R_{(1)}] + \sum_{j=2}^{k} \bE \left[R_{(j)} - R_{(j-1)}\right] \notag\\
&= \sum_{j=1}^{k} \frac{1}{(n+1-j)\mu} \notag\\
&= \frac{1}{\mu}(\mathbf{H}(n) - \mathbf{H}(n-k)).\label{eq:e_R_k}
\end{align}

Next, we analyze the second term of the right-hand side of Eq.~\eqref{eq:e_aoi} and want to show the following:
\begin{equation}
\label{eq:e_Delta_i}
\bE \left [\min_{i \in K} \Delta_i(s) \right] = \frac{1}{k\lambda}.
\end{equation}
Note that the updating process at each server is a Poisson process with rate $\lambda$ and is \emph{i.i.d.} across the servers.
Hence, the inter-update time for each server is exponentially distributed with mean $1/\lambda$. Due to the
memoryless property of the exponential distribution, the AoI at each server has the same distribution as the 
inter-update time, i.e., random variable $\Delta_i(s)$ is also exponentially distributed with mean $1/\lambda$ 
and is \emph{i.i.d.} across the servers \cite{nelson13}. Therefore, random variable $\min_{i \in K} \Delta_i(s)$
is the minimum of $k$ \emph{i.i.d.} exponential random variables with mean $1/\lambda$, which is also exponentially
distributed with mean $\frac{1}{k\lambda}$. This implies Eq.~\eqref{eq:e_Delta_i}.

Combining Eqs.~\eqref{eq:e_R_k} and \eqref{eq:e_Delta_i}, we complete the proof.
\end{proof}

\emph{Remark.}
The above analysis indeed agrees with our intuition: while the expected total waiting time for receiving the first $k$ responses (i.e., Eq.~\eqref{eq:e_R_k}) 
is a monotonically increasing function of $k$, the expected AoI of the freshest information among these $k$ responses 
(i.e., Eq.~\eqref{eq:e_Delta_i}) is a monotonically decreasing function of $k$.

\section{Optimal Replication Scheme}\label{sec:opt}
In this section, we will exploit the aforementioned tradeoff and focus on answering the second question 
we discussed at the end of Section~\ref{sec:model}.
Specifically, we aim to find the optimal number of responses to wait for in order to minimize the expected 
AoI at the user's side.

Using the analytical result of Theorem~\ref{thm:e_aoi}, we rewrite the optimization problem in Eq.~\eqref{eq:opt} as:
\begin{equation}
\label{eq:opt2}
 \min_{k \in \{1,2,\dots, n\}}   \bE [\Delta(k)] = \frac{1}{\mu}(\mathbf{H}(n) - \mathbf{H}(n-k)) + \frac{1}{k\lambda}.
\end{equation}
Let $k^*$ be an optimal solution to Eq.~\eqref{eq:opt2}.
We state the main result of this section in Theorem~\ref{thm:opt}.

\begin{theorem}
\label{thm:opt} 
An optimal solution $k^*$ can be computed as:
\begin{equation}
k^* = \min \left \{ \left \lceil \frac{2\mu n}{\sqrt{(\lambda + \mu)^2+4\lambda\mu n}+\lambda+\mu} \right \rceil, n  \right \}.
\end{equation}
\end{theorem}

\begin{proof}
We first define $D(k)$ as the difference of the expected AoI between the $(n,k+1)$ and $(n,k)$ 
replication schemes, i.e., $D(k) \triangleq \Delta(k+1) - \Delta(k)$ for any $k \in \{1,2,\dots,n-1\}$.
From Eq.~\eqref{eq:e_aoi_formula}, we have that for any $k \in \{1,2,\dots,n-1\}$,
\begin{equation}\label{eq:diff}
D(k) = \frac{1}{(n-k)\mu} - \frac{1}{k(k+1)\lambda}.
\end{equation}
It is easy to see that $D(k)$ is a monotonically increasing function of $k$.

We now extend the domain of $D(k)$ to the set of positive real numbers and want to find $k^{\prime}$
such that $D(k^{\prime}) = 0$. With some standard calculations and dropping the negative solution, 
we derive the following:
\begin{equation}
k^{\prime} = \frac{2\mu n}{\sqrt{(\lambda + \mu)^2+4\lambda\mu n}+\lambda+\mu}.
\end{equation}
Next, we discuss two cases: (i) $k^{\prime}>n$ and (ii) $0 < k^{\prime} \le n$.

In Case (i), we have $k^{\prime}>n$. This implies that $D(k)=\Delta_i(k+1)-\Delta_i(k)<0$ for all $k \in \{1,2,\dots,n\}$ due to the fact that
$D(k)$ is monotonically increasing. Hence, the expected AoI $\Delta(k)$ is a monotonically decreasing function for $k \in \{1,2,\dots,n\}$.
Therefore, $k^*=n$ must be an optimal solution.

In Case (ii), we have $0< k^{\prime} \le n$. We consider two subcases: $k^{\prime}$ is an integer in $\{1,2,\dots,n\}$ and $k^{\prime}$ is not an integer.

If $k^{\prime}$ is an integer in $\{1,2,\dots,n\}$, we have $D(k)=\Delta(k+1)-\Delta(k)<0$ for $k \in \{1,2,\dots,k^{\prime}-1\}$ and $D(k)=\Delta(k+1)-\Delta(k)>0$ 
for $k\in\{k^{\prime}+1,\dots,n\}$.
Hence, the expected AoI $\Delta(k)$ is first decreasing (for $k \in \{1,2,\dots,k^{\prime}-1\}$) and then increasing (for $k\in\{k^{\prime}+1,\dots,n\}$).
Therefore, there are two optimal solutions: $k^*=k^{\prime}$ and $k^*=k^{\prime}+1$ since $\Delta(k^{\prime}+1) = \Delta(k^{\prime})$ (due to $D(k^{\prime})=0$).

If $k^{\prime}$ is not an integer, we have $D(k)=\Delta(k+1)-\Delta(k)<0$ for $k \in \{1,2,\dots,\lfloor k^{\prime} \rfloor\}$ and $D(k)=\Delta(k+1)-\Delta(k)>0$ 
for $k\in\{\lceil k^{\prime} \rceil,\dots,n\}$. 
Hence, the expected AoI $\Delta(k)$ is first decreasing (for $k \in \{1,2,\dots,\lfloor k^{\prime} \rfloor \}$) and then increasing (for $k\in\{\lceil k^{\prime} \rceil,\dots,n\}$).
Therefore, $k^* = \lceil k^{\prime} \rceil$ must be an optimal solution.

Combining two subcases, we have $k^* = \lceil k^{\prime} \rceil$ in Case (ii). Then, combining Cases (i) and (ii), we have 
$k^* = \min \{\lceil k^{\prime} \rceil, n\} = \min \left \{ \left \lceil \frac{2\mu n}{\sqrt{(\lambda + \mu)^2+4\lambda\mu n}+\lambda+\mu} \right \rceil, n  \right \}$.
\end{proof}

\emph{Remark.} There are two special cases that are of particular interest: 
waiting for the first response only (i.e., $k^*=1$) and waiting for all the responses (i.e., $k^*=n$).
In Corollary~\ref{cor:specialcase}, we provide a sufficient and necessary condition for each of these two special cases.

\begin{corollary}
\label{cor:specialcase}
%An optimal solution satisfies the following: 
(i) $k^*=1$ is an optimal solution if and only if $\lambda \geq \frac{\mu(n-1)}{2}$;
(ii) $k^*=n$ is an optimal solution if and only if $\lambda \leq \frac{\mu}{n(n-1)}$.
\end{corollary}

\begin{proof}
The proof follows straightforwardly from Theorem~\ref{thm:opt}. 
A little thought gives the following: $k^*=1$ is an optimal solution if and only if $D(1) \ge 0$. 
Solving $D(1) = \frac{1}{(n-1)\mu} - \frac{1}{2\lambda} \ge 0$ gives $\lambda \geq \frac{\mu(n-1)}{2}$.
Similarly, $k^*=n$ is an optimal solution if and only if $D(n-1) \le 0$. 
Solving $D(n-1) = \frac{1}{\mu} - \frac{1}{n(n-1)\lambda} \le 0$ gives $\lambda \leq \frac{\mu}{n(n-1)}$.
\end{proof}

\emph{Remark.} The above results agree well with the intuition.
For a given number of servers, if the mean inter-update time is much smaller than the mean response 
time (i.e., $\lambda \gg \mu$), then all the servers have frequent updates and thus, the difference of the 
freshness levels among the servers is small. In this case, it is not beneficial to wait for more responses.
On the other hand, if the mean inter-update time is much larger than the mean response time (i.e., $\lambda \ll \mu$), 
then one server may possess fresher information than another server. In this case, it is worth waiting for 
more responses, which leads to a significant gain in the AoI reduction.

\begin{figure*}[t]
\centering
\begin{subfigure}[b]{0.31\linewidth}
	\centering
	\includegraphics[width=0.9\textwidth]{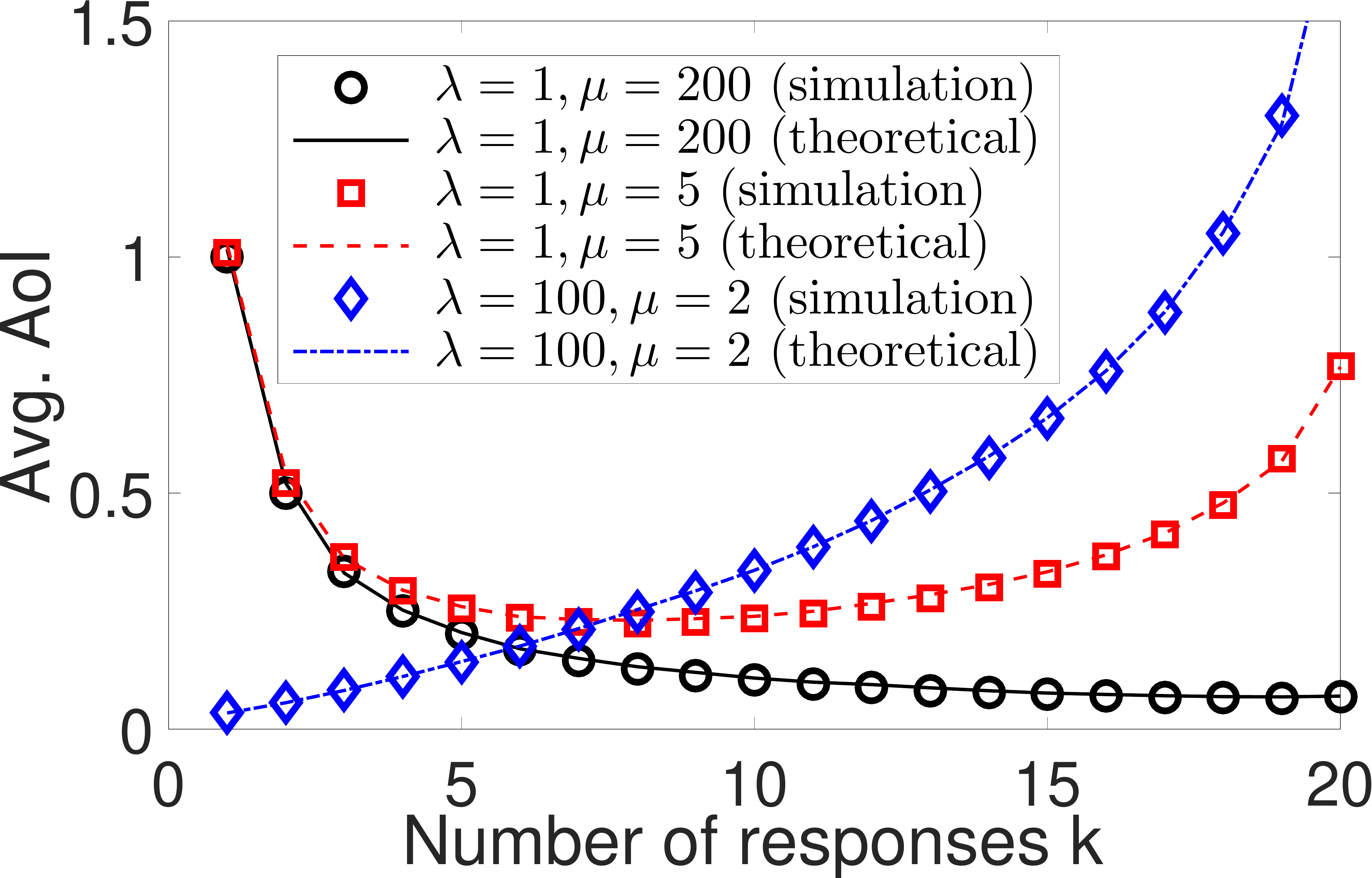}
	\caption{Exponential response time}
	\label{fig:exp}
\end{subfigure}
%\quad
\begin{subfigure}[b]{0.31\linewidth}
	\centering
	\includegraphics[width=0.9\textwidth]{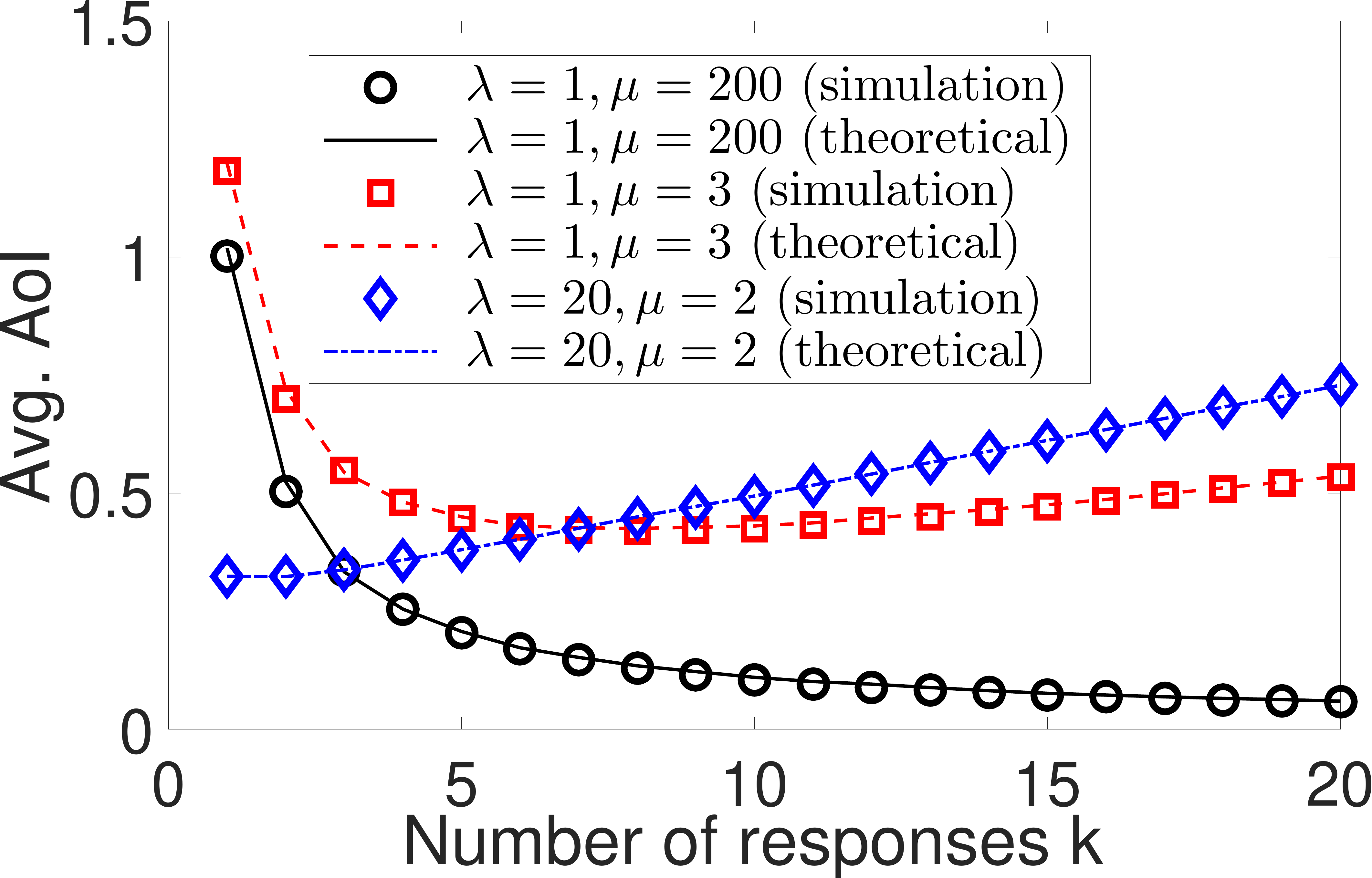}
	\caption{Uniform response time}
	\label{fig:uni}
\end{subfigure}
%\quad
\begin{subfigure}[b]{0.33\linewidth}
	\centering
	\includegraphics[width=0.9\textwidth]{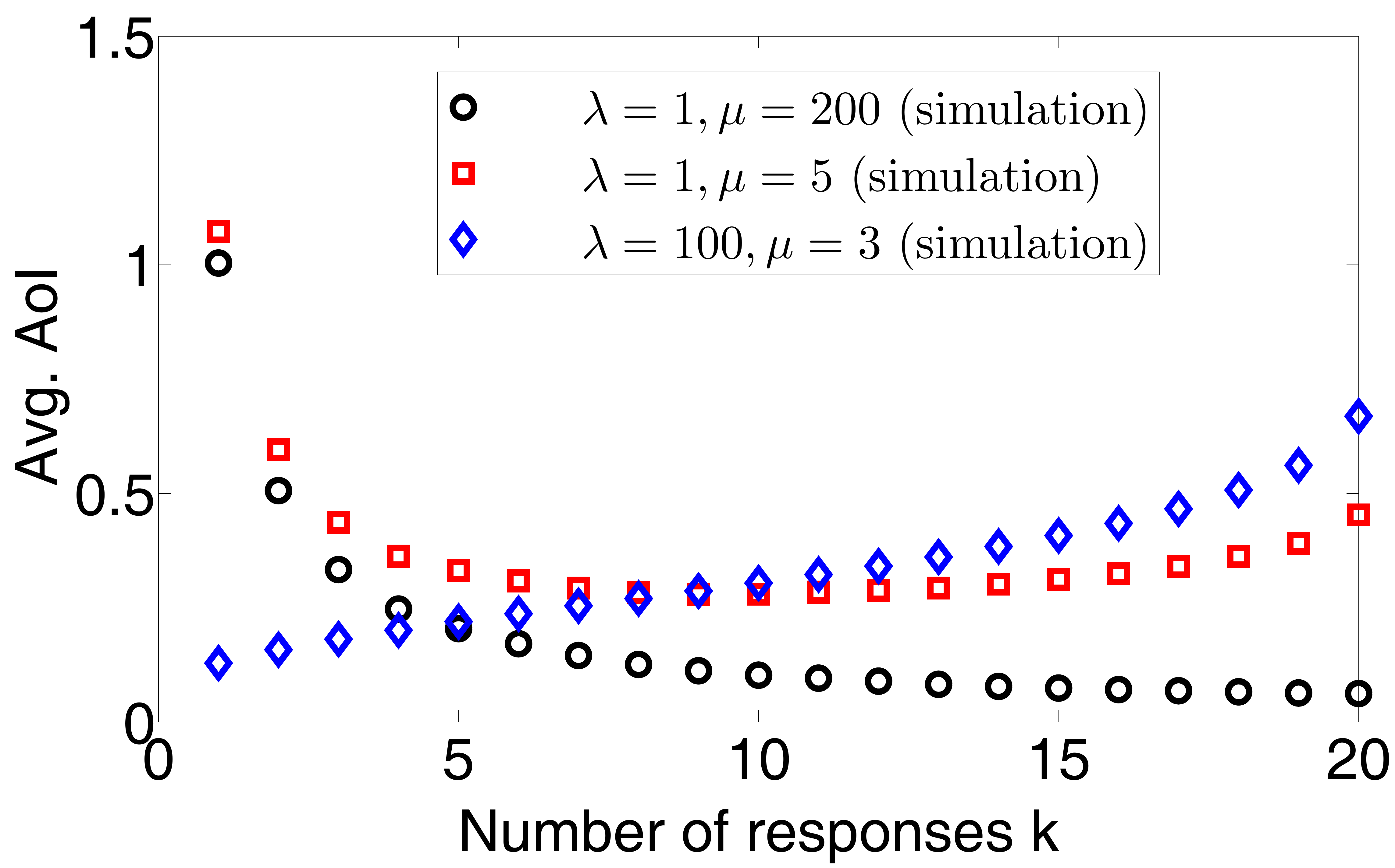}
	\caption{Gamma response time}
	\label{fig:gamma}
\end{subfigure}
\caption{Simulation results of average AoI vs. the number of responses $k$ for three different types of response time distributions.}
\label{fig:sim}
\end{figure*}

\section{Extensions}\label{sec:extensions}
In this section, we discuss some extensions of our work. 

\textbf{Replication scheme.}
So far, we have only considered the $(n,k)$ replication scheme. 
%This scheme requires the user to send a replicated request to every server, 
%which may incur a large overhead when there are a large number of servers (i.e., $n$ is large). 
\blue{One limitation of this scheme is that it requires the user to send a replicated request to every 
server, which may incur a large overhead when there are a large number of servers (i.e., when $n$ is large). }
Instead, a more practical scheme would be to send the replicated requests to a subset of servers.
Hence, we consider the $(n,m,k)$ replication schemes, under which the user sends a replicated 
request to each of the $m$ servers that are randomly and uniformly chosen from the $n$ servers, 
and waits for the first $k$ responses, where $m \in \{1,2,\dots,n\}$ and $k \in \{1,2,\dots,m\}$. 
Making the same assumptions as in Section~\ref{sec:model}, we can derive the expected AoI at 
the user's side in a similar manner. Specifically, reusing the proof of Theorem~\ref{thm:e_aoi} and
replacing $n$ with $m$ in the proof, we can show the following:
\begin{equation}
\bE[\Delta(k)] = \frac{1}{\mu}(\mathbf{H}(m) - \mathbf{H}(m-k)) + \frac{1}{k\lambda}.
\end{equation}

\textbf{Uniformly distributed response time.}
Note that our current analysis requires the memoryless property of the Poisson updating process. 
However, the analysis can be extended to the uniformly distributed response time. We make the 
same assumptions as in Section~\ref{sec:model}, except that the response time is now uniformly 
distributed in the range of $[a,a+h]$ with $a \ge 0$ and $h \ge 0$.
In this case, we have $\bE[R_{(k)}] = \frac{kh}{n+1} + a$ \cite{firstcourse}.
Since Eq.~\eqref{eq:e_Delta_i} still holds, from Eq.~\eqref{eq:e_aoi} we have
\begin{equation}
\label{eq:e_aoi_uniform}
\bE[\Delta(k)] = \frac{kh}{n+1} + a + \frac{1}{k\lambda}.
\end{equation}

Following a similar line of analysis to that in the proof of Theorem~\ref{thm:opt},
we can show that an optimal solution $k^*$ can be computed as:
\begin{equation}
k^* =\min \left\{ \left\lceil \frac{2(n+1)}{\sqrt{h^2\lambda^2+4h\lambda(n+1)} + h\lambda} \right \rceil, n \right\}.
\end{equation}

\section{Numerical Results}\label{sec:simulation}
In this section, we perform extensive simulations to evaluate the AoI performance in an information-update system with $20$ servers under the $(n,k)$ replication scheme. 
We first describe our simulation settings.
Throughout the simulations, the updating process at each server is assumed to be Poisson with rate $\lambda$ and is \emph{i.i.d.} across the servers. The user's request for the information is generated at time $s$, which is uniformly selected from the time interval $[0, T]$, where we set  $T = 10^6/\lambda$. This implies that each server has a total of $10^6$ updates on average. 
Next, we evaluate the AoI performance through simulations for three types of response time distribution: 
\emph{exponential}, \emph{uniform}, and \emph{Gamma}.
%The results are shown in Fig.~\ref{fig:sim}.\\
%\noindent\textbf{Exponential Response Time:} 
First, we assume that the response time is exponentially distributed with mean $1/\mu$.
Fig.~\ref{fig:exp} presents how the average AoI changes as the number of responses $k$ varies
in three representative setups, where each point represents an average of $10^3$ simulation runs.
We also include plots of our theoretical results (i.e., Eq.~\eqref{eq:e_aoi_formula}) for comparison.
A crucial observation from Fig.~\ref{fig:exp} is that the simulation results match perfectly with our theoretical results. 
In addition, we observe three different behaviors of the average AoI performance:
(i) If the inter-update time is much smaller than the response time (i.e., $\lambda = 100$, $\mu = 2$), 
then the average AoI increases as $k$ increases and thus, it is not beneficial to wait for more than one response. 
(ii) In contrast, if the inter-update time is much larger than the response time (i.e., $\lambda = 1$, $\mu = 200$), 
then the average AoI decreases as $k$ increases and thus, it is worth waiting for all the responses so as to achieve a smaller average AoI. 
(iii) When the inter-update time is comparable to the response time (i.e., $\lambda=1$, $\mu=5$),
then as $k$ increases, the AoI would first decrease and then increase. 
On the one hand, when $k$ is small, the freshness of the data at the servers dominates and thus, waiting for more responses helps reduce the average AoI. 
On the other hand, when $k$ becomes large, the total waiting time becomes dominant and thus, the average AoI increases as $k$ further increases.

%\textbf{Uniform Response Time:}
In Section~\ref{sec:extensions}, we discussed the extension of our theoretical results to the case of uniformly distributed response time.
Hence, we also perform simulations for the response time uniformly distributed in the range of $[\frac{1}{2\mu}, \frac{3}{2\mu}]$ with mean $1/\mu$. 
Fig.~\ref{fig:uni} presents the average AoI as the number of responses $k$ changes. In this scenario, the simulation results also
match perfectly with the theoretical results (i.e., Eq.~\eqref{eq:e_aoi_uniform}). 
Also, we observe a very similar phenomenon to that in Fig.~\ref{fig:exp} on how the average AoI varies as $k$ increases in three 
different simulation setups. 

\begin{figure*}[t!]
	\centering
	\begin{subfigure}[b]{0.3\linewidth}
		\includegraphics[width=0.9\textwidth]{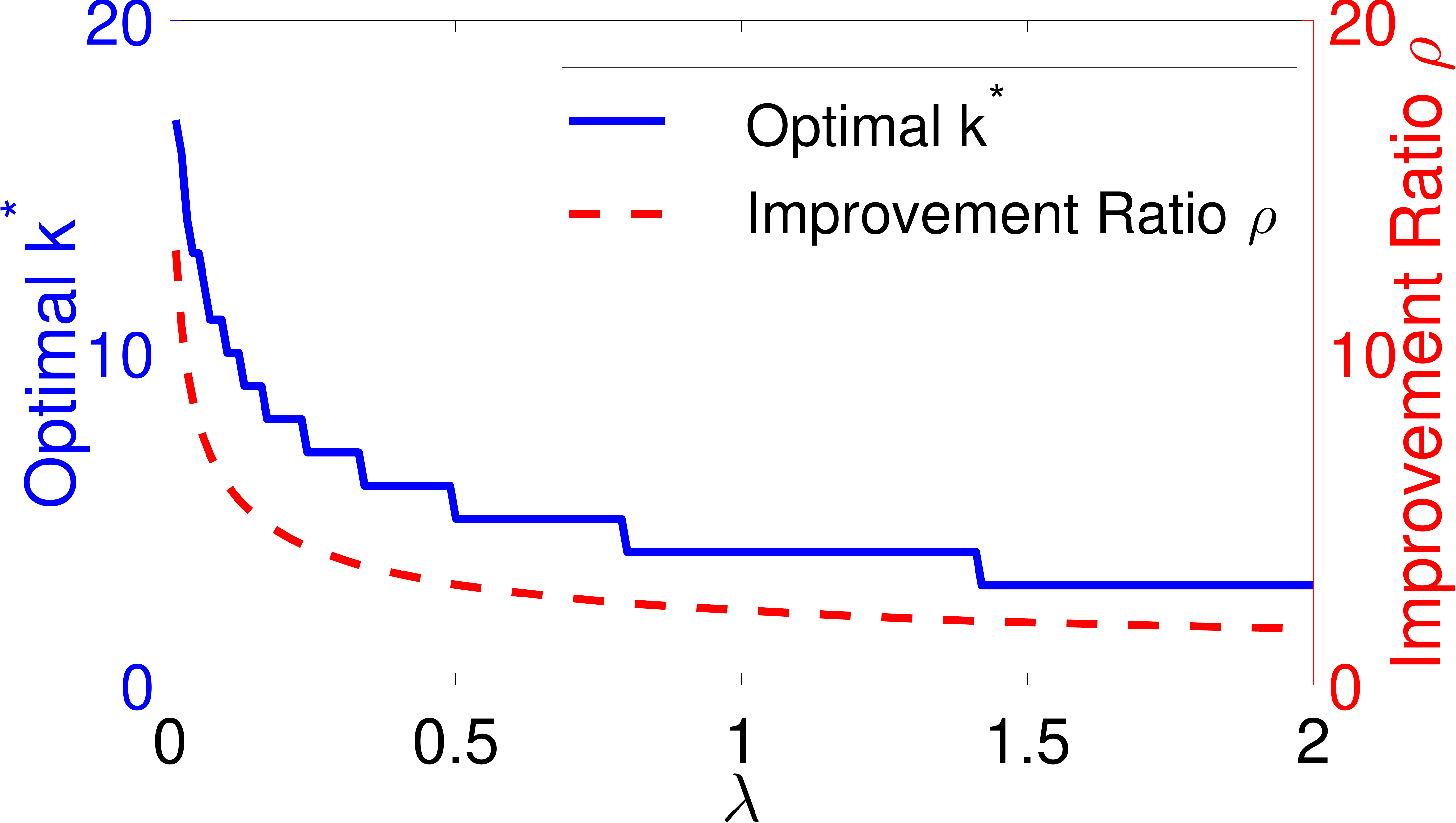}
		\caption{Impact of updating rate $\lambda$.}
		\label{subfig:l}
	\end{subfigure}
	\begin{subfigure}[b]{0.3\linewidth}
		\includegraphics[width=0.9\textwidth]{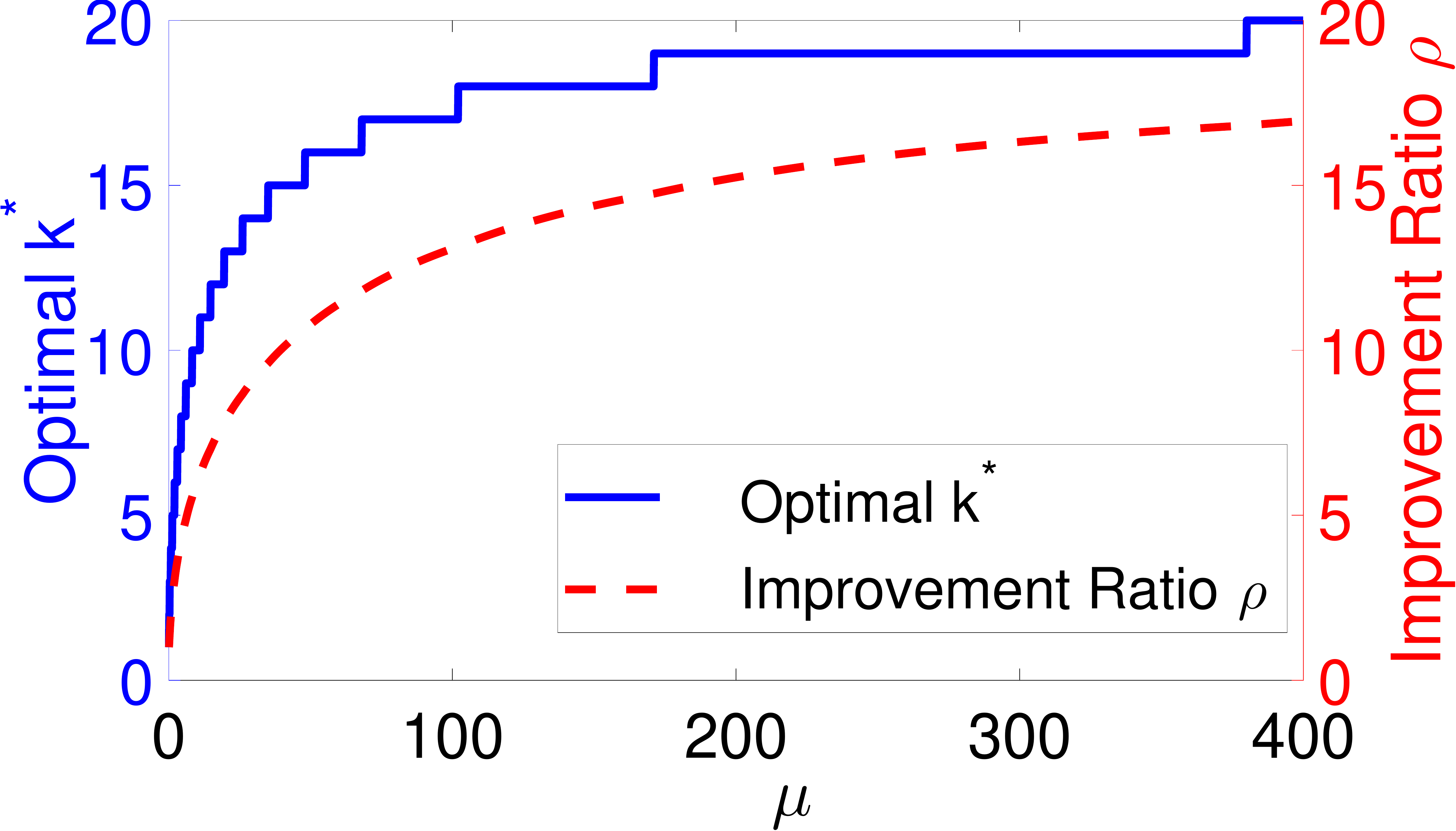}
		\caption{Impact of mean response time $1/\mu$.}
		\label{subfig:m}
	\end{subfigure}
	\begin{subfigure}[b]{0.3\linewidth}
		\includegraphics[width=0.9\textwidth]{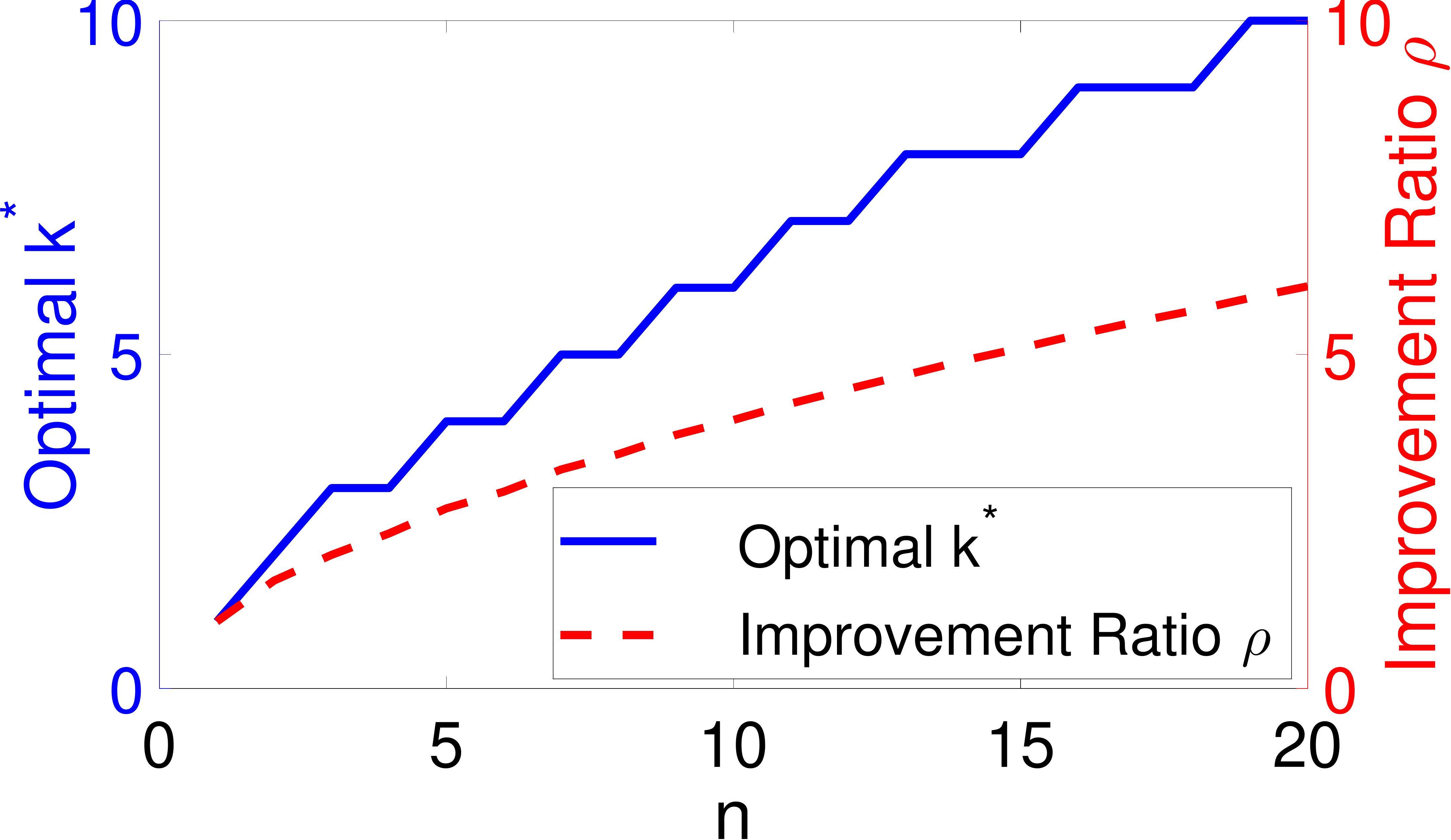}
		\caption{Impact of total number of servers $n$.}
		\label{subfig:n}
	\end{subfigure}
	\caption{Impact of the system parameters on the optimal $k^*$ and the corresponding improvement ratio.
	We consider the exponential distribution for the response time.
	In (a), we fix $\mu = 1, n = 20$; in (b), we fix $\lambda = 1, n = 20$; in (c), we fix $\lambda = 1, \mu = 10$.}
	\label{fig:impacts}
\end{figure*}

In addition, Fig.~\ref{fig:gamma} presents the simulation results for the response time with Gamma distribution, 
which can be used to model the response time in relay networks\cite{najm16}. 
Specifically, we consider a special class of the Gamma($r,\theta$) distribution that is the sum 
of $r$ \emph{i.i.d.} exponential random variables with mean $\theta$  (which is also called the Erlang distribution).
Then, the mean response time $1/\mu$ is equal to $r\theta$. We fix $r=5$ in the simulations.
Although we are unable to derive analytical results in this case, the observations are similar to that 
under the exponential and uniform distributions. 

%Fig.~\ref{fig:gamma} shows the AoI performance under the Gamma distribution of response time. 
%We have similar observations with cases with both exponential and uniform distributions. 
%However, the optimal number of responses $k^*$ varies under different simulation setups. 
%This is expected from our theoretical derivations. 
%Therefore, it would be interesting to develop a robust algorithm to determine the optimal number of 
%responses without the knowledge of both distributions of information updating time and response time.

%as in Fig.~\ref{fig:exp} 

%Then, we change the value of ($\lambda, \bar{R}, k$) to get the age performance under different scenarios.
%As we can see from Fig.~\eqref{fig:uni}, the trageoff between data freshness and waiting time still exists under uniform distribution.
%We can see three types of curves but they are different from Fig.~\ref{fig:exp}.
%Particularly, when the mean response time is small ($\lambda = 20, \bar{R} = 1/2$), the average age will increase linearly as $k$ increases.
%This result complies with our theoretical result as shown by Eq.~\eqref{eq:uni}.
%It is easy to see that the expected AoI is approximately a linear function of $k$ when $\lambda$ or $k$ is large enough.\\

%\subsection{Impact of System Parameters}
Finally, we investigate the impact of the system parameters (the updating rate, the mean response time, 
and the total number of servers) on the optimal number of responses $k^*$ and the \emph{improvement ratio}, 
defined as $\rho \triangleq \mathds{E}[\Delta(1)]/\mathds{E}[\Delta(k^*)]$.
The improvement ratio captures the gain in the AoI reduction under the optimal scheme 
compared to a naive scheme of waiting for the first response only. 
%The plots are provided in Fig.~\ref{fig:impacts}.
%\noindent\textbf{Impact of updating rate:} 

Fig.~\ref{subfig:l} shows the impact of the updating rate $\lambda$. 
We observe that the optimal number of responses $k^*$ decreases as $\lambda$ increases. 
This is because when the updating rate is large, the AoI diversity at the servers is small. 
In this case, waiting for more responses is unlikely to receive a response with much fresher 
information. Therefore, 
the optimal scheme will simply be a naive scheme that waits only for the first response when the updating rate is relatively large (e.g., $\lambda=2$). 
%We can observe from Fig.~\ref{subfig:l} that the optimal solution $k^*$ decreases as the servers update more frequently, i.e., as $\lambda$ becomes larger. This implies that if each server updates its information fast, then it is better to wait for fewer responses to maximize the AoI performance and thus the naive solution will become optimal as $\lambda$ further increases. 
%
%This also matches our analysis in Section~\ref{sec:opt} on how the information updating rate influences the optimal $k^*$. 
%
%\noindent\textbf{Impact of mean response time:} 
Fig.~\ref{subfig:m} shows the impact of the mean response time $1/\mu$. 
We observe that the optimal number of responses $k^*$ increases as $\mu$ increases.
This is because when $\mu$ is large (i.e., when the mean response time is small),
the cost of waiting for additional responses becomes marginal and thus, waiting for more 
responses is likely to lead to the reception of a response with fresher information.
%From Fig.~\ref{subfig:m}, we can observe that when the mean response time is small, i.e., $\mu$ becomes large, the cost of waiting for more responses is small and thus it is always beneficial to wait for more responses to reduce the age of information. More performance gains can be achieved as $\mu$ further increases.
%On the other hand, the optimal solution is also influenced by the response time.
%When response time is short, i.e., $\mu$ becomes large,  the price of waiting for more responses decreases. 
%In such cases, we can wait for more responses to get a fresher data.
% This is verified by Fig.~\ref{subfig:m}.
% We can see the optimal solution increases as $\mu$ increases.
%
%\noindent\textbf{Impact of total number of servers:} 
Fig.~\ref{subfig:n} shows the impact of the total number of servers $n$.
We observe that both the optimal number of responses $k^*$ and
the improvement ratio increase with $n$.
This is because an increased number of servers leads to more diversity gains
both in the AoI at the servers and in the response time.

\section{Conclusion}\label{sec:conclusion}
In this paper, we introduced a new Pull model for studying the AoI minimization problem under the replication schemes.
Assuming Poisson updating process and exponentially distributed response time, we derived the closed-form expression 
of the expected AoI at the user's side and provided a formula for computing the optimal solution.
Not only did our work reveal a novel tradeoff between different levels of information freshness and different response 
times across the servers, but we also demonstrated the power of waiting for more than one response in minimizing 
the expected AoI at the user's side.
An interesting direction for future work would be to develop dynamic replication schemes that do not require the knowledge
of the updating process and the response time distribution.

%In this paper, we propose a new \emph{Pull Model} of communication systems for the first time.
%We consider the metric, Age-of-Information, of this new model and propose to use redundant requests to minimize the age performance.
%Our analysis shows that the use of redundancy in pull model is very different from other scenarios and it has a unique tradeoff between response freshness and waiting time.
%We provide extensive analysis and a closed-form expression of the expected AoI of such systems in this paper.
%Besides that, we formulate the optimization problem of using redundant requests to minimize AoI and provide its optimal solution.  
%(\blue{reference 15, author names are omitted since they are the same as 14. do we need to modify it manually?})

\bibliographystyle{IEEEtran}
\bibliography{AoI}

\end{document}